\documentclass[aps,preprint,pra,showpacs]{revtex4}%

\usepackage[centertags]{amsmath}
\usepackage{amsthm}
\usepackage{newlfont}

\usepackage{graphicx}
\usepackage[pdftex]{color}
\usepackage{amssymb}
\usepackage{amsfonts}
\usepackage{amsmath}%
\usepackage{color}
\providecommand{\U}[1]{\protect\rule{.1in}{.1in}}

\def\be{\begin{equation}}
\def\ee{\end{equation}}
\def\bea{\begin{eqnarray}}
\def\eea{\end{eqnarray}}

\newtheorem{thm}{Theorem}[section]

\newtheorem{defn}[thm]{Definition}

\begin{document}
\title{Entangled Histories vs. the Two-State-Vector Formalism - Towards a Better Understanding of Quantum Temporal Correlations}
\author{Marcin Nowakowski\footnote{Electronic address: marcin.nowakowski@pg.edu.pl}}
\affiliation{Faculty of Applied Physics and Mathematics, ~Gdansk University of Technology, 80-952 Gdansk, Poland}
\affiliation{National Quantum Information Center, Andersa 27, 81-824 Sopot, Poland}
\author{Eliahu Cohen}
\affiliation{Physics Department, Centre for Research in Photonics, University of Ottawa,
Advanced Research Complex, 25 Templeton, Ottawa ON Canada, K1N 6N5}
\author{Pawel Horodecki}
\affiliation{National Quantum Information Center, Andersa 27, 81-824 Sopot, Poland}
\affiliation{Faculty of Applied Physics and Mathematics, ~Gdansk University of Technology, 80-952 Gdansk, Poland}


\begin{abstract}
The Two-State-Vector formalism and the Entangled Histories formalism are attempts to better understand quantum correlations in time. Both formalisms share some similarities, but they are not identical, having subtle differences in their interpretation and manipulation of quantum temporal structures. However, the main objective of this paper is to prove that, with appropriately defined scalar products, both formalisms can be made isomorphic. We show, for instance, that they treat operators and states on equal footing leading to the same statistics for all measurements. In particular, we discuss the topic of quantum correlations in time and show how they can be generated and analyzed in a consistent way using these formalisms. Furthermore, we elaborate on a novel behavior of quantum histories of evolving multipartite systems which do not exhibit global non-local correlations in time but nevertheless can lead to entangled reduced histories characterizing evolution of an arbitrarily chosen sub-system.
\end{abstract}

\pacs{03.65.Ud, 03.65.Ta, 03.67.Hk}
\maketitle

\section{Introduction}\label{sec1}

Quantum correlations, both spatial and temporal are unique. They create a clear distinction
between classical, quantum and post-quantum theories. For analyzing quantum
correlations in time, two apparently different formalisms have been proposed (among several others) - Multiple-Time States by Aharonov {\it et al.} \cite{AAD,MTS1} as part of the Two-State-Vector formalism (TSVF) \cite{ABL,Properties,TSVFR,TTI}, and the Entangled Histories (EH) approach due to Cotler and Wilczek \cite{Cot}. These two formalisms not only provide a richer
notion of the history of a quantum system, but also allow to study the intricate temporal correlations it entails.
Both approaches have shown to be very fruitful: The TSVF led to surprising effects within pre- and postselected systems (e.g. \cite{Paradoxes,Pigeon,Dis}), time travel thorough post-selected teleportation \cite{Lloyd1,Lloyd2}, a novel notion of quantum time \cite{Each}, new results regarding quantum state tomography \cite{MTS2} and a better understanding of processes with indefinite causal order \cite{NewSandu}, while the Entangled Histories approach led to Bell tests for histories \cite{WC3} and monogamy of quantum entanglement in time \cite{Now2}. Together, they have been recently used for analyzing the final state proposal in black holes \cite{CN}.

In the following section we revisit these two formalisms from a modern perspective, analyzing and extending their main features.  The united approach we develop is used for deriving new results regarding quantum correlations in time within Sec. III and the Appendix. Although both formalisms are not completely equivalent due to minor conceptual differences and their different definitions of inner products, we prove in Sec. IV that they can be made isomorphic under certain conditions. We conclude with some general remarks regarding the nature of quantum time.

\section{Multiple-Time States vs. Entangled Histories}

We shall begin with the mathematical construction of the two formalisms. Both seek to encode the evolution of the system through time in a complete, yet still compact way. They do so in quite a similar manner, but with a subtle difference.

Multiple-Time States (MTS) extend the standard quantum mechanical state by allowing its simultaneous description in several different moments. Moreover, such a multiple-time state may encompass both forward- and backward-evolving states on equal footing. The motivation for allowing so is restoring time-symmetry in quantum mechanics, which is apparently lost upon collapse \cite{ABL}. Thus, MTS represent all instances of collapse (i.e. those moments in time when the quantum state coincided with an eigenstate of some measured operator) and allow them to evolve both forward and backward in time. This evolution backwards in time can be understood literally (giving rise to the Two-Time Interpretation \cite{TTI}), but this is not necessary, it can be simply regarded as a mathematical feature of the formalism (which is, in fact, equivalent to the standard quantum formalism \cite{TSVFR}). MTS live in a tensor product of Hilbert spaces $\mathcal{H}$ admissible at those various instances of time ($t_1<...<t_n$) denoted by \cite{MTS1}
\begin{equation}\label{HMTF}
\mathcal{H}=\mathcal{H}_{t_{n}}^{(\cdot)}\otimes...\otimes\mathcal{H}_{t_{k+1}}^{\dagger}\otimes\mathcal{H}_{t_{k}}\otimes\mathcal{H}_{t_{k-1}}^{\dagger}\otimes...\otimes\mathcal{H}_{t_{1}}^{(\cdot)},
\end{equation}
where a dagger means the corresponding Hilbert space consists of states which evolve backwards in time. The initial and final Hilbert spaces might be daggered or not (this is denoted by a ``$\cdot$'' superscript). All Hilbert spaces containing either (forward-evolving) kets or (backward-evolving) bras are alternating to allow a time-symmetric description at any intermediate moment. An example of (a separable) MTF would be: $_{t_4}\langle z^+||x^-\rangle_{t_3~t_2}\langle y^-||x^+\rangle_{t_1} \in \mathcal{H}_{t_{4}}^\dagger\otimes\mathcal{H}_{t_{3}}\otimes\mathcal{H}_{t_{2}}^\dagger\otimes\mathcal{H}_{t_{1}}$. This multiple-time state represents an initial eigenstate of the Pauli-X operator evolving forward in time from $t_1$ until collapse into an eigenstate of the Pauli-Y operator occurs at time $t_2$. Later on, at time $t_3$ the system is projected again onto a different eigenstate of the Pauli-X operator. Finally at $t_4$ the system is measured in the Z basis, and the resulting eigenstate evolves backward in time. In the following we will focus on two-time states (sometimes called two-states), which consist of a forward evolving state $|\psi_1\rangle_{t_{1}}$ and a backward evolving state $|\psi_2\rangle_{t_{2}}$ in the above form $_{t_1}\langle \psi_1| |\psi_2\rangle_{t_2}$ to achieve a richer description of a quantum system during the time interval $t_1\le t \le t_2$ \cite{TSVFR}. We would hence omit the obvious sub-indices $t_i$.

Given an initial state $|\Psi\rangle$ and a final state $\langle \Phi|$, the probability that an intermediate measurement of some hermitian operator $A$ will result in the eigenvalue $a_n$ is given by the ABL formula \cite{ABL}

\begin{equation} \label{ABL}
p(A=a_n)=\frac{1}{N}|\langle\Phi|U_{2}P_{n}U_{1}|\Psi\rangle|^2,
\end{equation}
where $U_i$ represent unitary evolution, the operator $P_n$ projects on $|a_n\rangle$ and
\begin{equation}
N \equiv \sum_k |\langle\Phi|U_{2}P_{k}U_{1}|\Psi\rangle|^2.
\end{equation}
This probability rule is important in that it uses the information available through the final state in a way which is manifestly time-symmetric.

Let us examine now the entangled histories (EH) approach. Its predecessor, the decoherent histories theory (or consistent histories theory) has a long tradition \cite{Hartle1, Hartle2, Gri84, Rep1, Rep2} and is built on the grounds of the well-known and broadly applied Feynman's path integral theory for calculation of probability amplitudes of quantum processes, especially in quantum field theory. It is presented also as a generalization of quantum mechanics applied to closed systems such as the universe as a whole and discussed as a necessary element of future quantum gravity theory \cite{Hartle1}. The EH formalism extends the concepts of the consistent histories theory by allowing for complex superposition of histories. Contrary to the TSVF, in the EH approach there is no notion of backwards evolution and hence all Hilbert spaces and all states are treated on equal footing \cite{Cot} \textit{keeping time-symmetry}. A history state is understood as an element in $\text{Proj}(\mathcal{H})$, spanned by projection operators from $\mathcal{H}$ to $\mathcal{H}$, where here $\mathcal{H}=\mathcal{H}_{t_{n}}\odot...\odot\mathcal{H}_{t_{1}}$. The $\odot$ symbol, which we use to comply with the current literature, stands for sequential tensor products, and has the same meaning as the above $\otimes$ symbol. A typical (separable) history state would be therefore denoted as: $|H)=P_{n}\odot...\odot P_{1} \in \text{Proj}(\mathcal{H}_{t_{n}}\odot...\odot\mathcal{H}_{t_{1}})$, where a physical system is understood to have a property $P_{i}$ at time $t_i$. As an example, one can take a history $|H)=[z^+]\odot[x^-]\odot[y^-]\odot[x^+]=[|z^+\rangle\langle z^+|]\odot[|x^-\rangle\langle x^-|]\odot[|y^-\rangle\langle y^-|]\odot[|x^+\rangle\langle x^+|]$ for a spin-$\frac{1}{2}$ particle being in an eigenstate of the Pauli-X operator at time $t_1$, in an eigenstate of the Pauli-Y operator at time $t_2$, and so on. Within this formalism one also defines the unitary bridging operators $\mathcal{T}(t_j,t_i):\mathcal{H}_{t_i}\rightarrow\mathcal{H}_{t_j}$ evolving the states between instances of time, and having the following properties: $\mathcal{T}(t_j,t_i)=\mathcal{T}^{\dagger}(t_i,t_j)$ and $\mathcal{T}(t_j,t_i)=\mathcal{T}(t_j,t_{j-1})\mathcal{T}(t_{j-1},t_i)$.

Following the consistent histories theory, the alternatives at a given instance of time form an exhaustive orthogonal set of projectors  $\sum_{\alpha_{x}}P_{x}^{\alpha_{x}}=\mathbb{I}$ and for the sample space of entangled histories $|H^{\overline{\alpha}})=P_{n}^{\alpha_{n}}\odot P_{n-1}^{\alpha_{n-1}}\odot\ldots\odot P_{1}^{\alpha_{1}}\odot P_{0}^{\alpha_{0}}$ ($\overline{\alpha}=(\alpha_{n}, \alpha_{n-1},\ldots, \alpha_{0})$), there exists $c_{\overline{\alpha}} \in \mathbb{C}$ such that $\sum_{\overline{\alpha}}c_{\overline{\alpha}}|H^{\overline{\alpha}})=\mathbb{I}$.

The formalism of consistent histories introduces also the chain operator $K(|H^{\overline{\alpha}}))$, which can be directly associated with a time propagator of a given quantum process:
\begin{equation}
K(|H^{\overline{\alpha}}))=P_{n}^{\alpha_n}\mathcal{T}(t_{n},t_{n-1}) P_{n-1}^{\alpha_{n-1}}\ldots\mathcal{T}(t_{2},t_{1})P_{1}^{\alpha_1}\mathcal{T}(t_{1},t_{0})P_{0}^{\alpha_0}.
\end{equation}
This operator plays a fundamental role in measuring a weight of any history $|H^{\alpha})$:
\begin{equation}
W(|H^{\alpha}))=TrK(|H^{\alpha}))^{\dagger}K(|H^{\alpha})),
\end{equation}
which can be interpreted as a realization probability of a history (see also the relation with sequential weak values \cite{GC}).

The chain operator could be used also for defining an inner semi-definite product for any two histories $|H)$ and $|Y)$ belonging to the same consistent family of histories:

\begin{equation}
(H|Y)_K=Tr[K^\dag(|H))K(|Y))],
\end{equation}
whose role will be discussed further in this paper.
Recent years have encompassed also an extensive discussion regarding the so-called consistency or decoherence of allowed histories \cite{Hartle1, Hartle2}, which is directly related to the degree of interference between pairs of histories within the set of histories. The consistent histories framework assumes that the family
of histories is consistent, i.e. one can associate with a union of histories a weight equal to the sum of weights
associated with particular histories included in the union. The histories belonging to the consistent family should also meet the strong consistency condition: $(H^{\alpha}|H^{\beta})_K=0$ for $\alpha\neq\beta$, although it is an open question whether weakened variants of this condition suffice for preservation of the probability distribution over a set of allowed histories and the orthonormality of the consistent family \cite{WC1, WC2, Hartle1, Hartle2}.

\section{Non-local correlations in time}
Both formalisms lead to a fundamental discussion about non-local correlations in time introducing mathematical structures which make these considerations well-founded.
Following the discussion about MTS, we can consider a quantum state exhibiting quantum entanglement in time with two times $t_2>t_1$, in TSVF representation:
\begin{equation}\label{state1}
\mathcal{H}_{t_{2}}^{\dagger}\otimes\mathcal{H}_{t_1}\ni|\Psi_{t_{2}t_{1}}\rangle\rangle=\alpha_0 \langle\Phi_0||\Psi_0\rangle+\alpha_1 \langle\Phi_1||\Psi_1\rangle,
\end{equation}
with non-zero probability amplitudes $\alpha_0, \alpha_1\in\mathbb{C}$ (some texts use a different notation for two-time states, but the current ``double ket'' will turn out to be useful later on).

The two-time state $|\Psi_{t_{2}t_{1}}\rangle\rangle$ can be represented as an entangled history in the EH formalism, assuming a unitary evolution $U(t_2,t_1)$ between the two times:
\begin{equation}\label{history1}
|H_{t_{2}t_{1}})=\alpha_{0}[|\Phi_0\rangle\langle\Phi_0|]\odot[|\Psi_0\rangle\langle\Psi_0|]+\alpha_{1}[|\Phi_1\rangle\langle\Phi_1|]\odot[|\Psi_1\rangle\langle\Psi_1|],
\end{equation}
and as we show next, there exists a natural isomorphism between the two state spaces: $\mathcal{H}_{t_{2}}^{\dagger}\otimes\mathcal{H}_{t_1}\ni|\Psi_{t_{2}t_{1}}\rangle\rangle \leftrightarrow|H_{t_{2}t_{1}})\in Proj(\mathcal{H}_{t_{2}})\odot Proj(\mathcal{H}_{t_1})$.

Let us slightly modify the evolution of the system, supposing that at some time $t_1<t<t_2$ an observable $A$ is measured and the measured eigenvalue $a_{n}$ corresponds to the projector $P_{n}$, then the probability of a realized eigenvalue $a_{n}$ is:

\begin{equation}
P(A=a_{n})=|\alpha_0 \langle\Phi_0|U_{2}P_{n}U_{1}|\Psi_0\rangle+\alpha_1 \langle\Phi_1|U_{2}P_{n}U_{1}|\Psi_1\rangle|^2,
\end{equation}
which is equivalent to the ABL formula in (\ref{ABL}) when $N=1$. Here the unitary operators correspond to $U_1=U(t,t_1)$ and $U_2=U(t_2,t)$, which can be treated as bridging operators for the new entangled history.
On the other hand, an action of the projection operator at time $t_1<t<t_2$ can be represented by a modification of the previous history with an ``injected'' operation between times $t_1$ and $t_2$:
\begin{equation}\label{enthistory}
|\tilde{H}_{t_{2}t_{1}})=\alpha_{0}[|\Phi_0\rangle\langle\Phi_0|]\odot P_{n}\odot[|\Psi_0\rangle\langle\Psi_0|]+\alpha_{1}[|\Phi_1\rangle\langle\Phi_1|]\odot P_{n}\odot[|\Psi_1\rangle\langle\Psi_1|],
\end{equation}
which leads to the same probability distribution of the realized history as can be found by employing the ABL formula. In this case it can be calculated as the aforementioned weight of the history \cite{WC1}:
\begin{equation}
Pr(|\tilde{H}_{t_{2}t_{1}}))=Tr K(|\tilde{H}_{t_{2}t_{1}}))^\dagger K(|\tilde{H}_{t_{2}t_{1}})).
\end{equation}

It is worth noting that the history $|\tilde{H}_{t_{2}t_{1}})$ is unnormalized (which is helpful for calculation of its realization probability by means of the $K(\cdot)$ operator). 
As in the spatial case, its normalization is straightforward with application of the inner product for calculation of probabilities: $|H^*_{t_{2}t_{1}})=\frac{|\tilde{H}_{t_{2}t_{1}})}{\sqrt{(\tilde{H}_{t_{2}t_{1}}|\tilde{H}_{t_{2}t_{1}})_K}}$.

It is interesting to find out how such an entangled state could be generated in reality (see the Appendix for an explicit example). Following the line of reasoning for TSVF in \cite{MTS1}, one can consider a bipartite system consisting of the system S and the ancilla A. The composite system is initiated at time $t_1$ in a state $\mathcal{H}_{S}\otimes\mathcal{H}_{A}\ni|\Psi\rangle=\lambda_0|\Psi_{0}0\rangle + \lambda_1|\Psi_{1}1\rangle$ and is post-selected at time $t_2$ in a state $|\Phi\rangle=\beta_0|\Phi_{0}0\rangle + \beta_1|\Phi_{1}1\rangle)$ that leads to the history $|H_{SA})=[\Phi]\odot[\Psi]$ and to the corresponding two-time state: $|\Psi_{SA}\rangle\rangle=\langle\Phi ||\Psi\rangle$.

Considering the bipartite system SA in a state $|\Psi_{SA}\rangle\rangle$, we perform a measurement only on the system $S$ leaving the ancilla undisturbed ($I_A$ stands for the identity operation acting on the ancilla A):
\begin{eqnarray}
P(A=a_{n})&=&  \nonumber\\
&=&|(\beta_{0}^*\langle\Phi_{0}0| + \beta_{1}^*\langle\Phi_{1}1|)[U_2\otimes I_{A}] [P_N\otimes I_{A}] [U_1\otimes I_{A}](\lambda_{0}|\Psi_{0}0\rangle + \lambda_{1} |\Psi_{1}1\rangle)|^2 \nonumber\\
&=&| \beta_{0}^*\lambda_{0}\langle\Phi_0|U_{2}P_{n}U_{1}|\Psi_0\rangle+ \beta_{1}^*\lambda_{1}\langle\Phi_1|U_{2}P_{n}U_{1}|\Psi_1\rangle|^2 \nonumber\\
\end{eqnarray}
exhibiting destructive interference for the ancillary system's orthogonal states at times $t_2$ and $t_1$. Due to this purely quantum effect and with appropriate adjustment of probability amplitudes ($\beta_{0}^*\lambda_{0}=\alpha_0$ and $\beta_{1}^*\lambda_{1}=\alpha_1$), represented in time by states of the form (\ref{state1}) and (\ref{history1}), we generate a probability distribution characteristic of entangled quantum states.

From an allowed history perspective (i.e. a history with non-zero realization probability), this evolution of the system and ancilla can be represented by a history:
\begin{equation}\label{history2}
|H_{SA})=\gamma[|\Phi\rangle\langle\Phi|]\odot [P_N\otimes I_{A}]\odot [|\Psi\rangle\langle\Psi|],
\end{equation}
where $\gamma$ stands for a normalization factor. On top of that, the evolution is determined by the bridging unitary operators: $B(t_2,t)=U_2\otimes I_{A}$ and $B(t,t_1)=U_1\otimes I_{A}$.


Yet we observe that the history (\ref{history2}) of the composite system $SA$ is a separable vector and as such does not indicate entanglement in time so one can raise a question: \textit{how can we derive quantum entanglement in time of the form (\ref{enthistory}) from this history (not breaking the concept of monogamy of entanglement and the general rules of contracting tensored spaces of quantum states)?}

It is simple to show for separable quantum states of a multipartite system that tracing out some of its sub-systems or contracting the global space of the system, in which the state lives, does not generate quantum entanglement.
In the case of the entangled history $|H_{SA})$ in time, tracing out the ancilla from the history requires application of the partial trace operation on a spatial component $A$ through all time frames. Yet, due to the bridging operators linking particular observation times, partial tracing in time \textit{cannot} be a mere analogy of the spatial case, i.e.:
\begin{equation}
Tr_A(\cdot)=\sum_{ijk} \langle i|_{t_2}\odot\langle j|_{t}\odot\langle k|_{t_1}  (\cdot) |i\rangle_{t_2}\odot|j\rangle_{t}\odot|k\rangle_{t_1}.
\end{equation}
It should be rather an operation capturing information about the evolution of the traced out sub-system during the time of an analyzed history:

\begin{defn}
For a history $|H_{SA})$ living in a space $\text{Proj}(\mathcal{H})=\text{Proj}(\mathcal{H}_{t_{n}}\otimes\dots\otimes\mathcal{H}_{t_{1}})$, a partial trace over its sub-system $A$ at all times $\{t_{j}\ldots t_{i+1} t_{i}\}$ $(j\geq i)$ is:
\begin{equation}\label{PartialTrace}
Tr_{A} |H_{SA})(H_{SA}|=\sum_{k=1}^{\dim\mathcal{F}} (e_{k}|H_{SA})_{K}(H_{SA}|e_{k})_{K},
\end{equation}
where $\mathcal{F}=\{|e_{k})\}$ creates an orthonormal consistent family \footnote{The family of consistent histories is a set of histories $\mathcal{F}=\{|H^{\overline{\alpha}})\}_{\overline{\alpha}=(\alpha_{n}, \alpha_{n-1},\ldots, \alpha_{0})}$, such that $\sum_{\overline{\alpha}}c_{\overline{\alpha}}|H^{\overline{\alpha}})=\mathbb{I}\; \text{for}\; c_{\overline{\alpha}} \in \mathbb{C}$ and any pair of histories from the set meets the consistency condition.} of histories on times $\{t_{j}\ldots t_{i+1} t_{i}\}$ of the system $A$ and the strong consistency condition for partial histories holds for the base histories, i.e. $(e_{i}|e_{j})_K=Tr[K(|e_{i}))^{\dag}K(|e_{j}))]=\delta_{ij}$.
\end{defn}
This definition is built in analogy to a partial trace operator over chosen  times of a history defined in \cite{Now2}. It becomes clearer in the example of history state $|H_{SA})$ why it is fundamental to look into evolution of the traced out parties.

Let us expand now a history $|H^*_{SA})$ being a simplification of $|H_{SA})$ (all internal phase factors are now equal):
\begin{eqnarray}
|H^*_{SA})&=&\gamma[|\Phi_{0}0\rangle\langle\Phi_{0}0| + |\Phi_{0}0\rangle\langle\Phi_{1}1|+|\Phi_{1}1\rangle\langle\Phi_{0}0|+|\Phi_{1}1\rangle\langle\Phi_{1}1|] \odot \\\nonumber
&&\odot[P_N\otimes I_{A}]\odot \\ \nonumber &&\odot[|\Psi_{0}0\rangle\langle\Psi_{0}0|+|\Psi_{0}0\rangle\langle\Psi_{1}1|+|\Psi_{1}1\rangle\langle\Psi_{0}0|+
|\Psi_{1}1\rangle\langle\Psi_{1}1|] \\\nonumber
\end{eqnarray}
with some normalization factor $\gamma$.
The history components of type: $[|\Phi_{0}0\rangle\langle\Phi_{1}1|]\odot[P_N\otimes I_{A}]\odot[|\Psi_{0}0\rangle\langle\Psi_{0}0|]$ vanish after tracing out the $A$-subsystem but then we find components of the history of the following type: $|h_{SA})=[|\Phi_{1}1\rangle\langle\Phi_{1}1|]\odot[P_N\otimes I_{A}]\odot[|\Psi_{0}0\rangle\langle\Psi_{0}0|]$ which cannot be realized (a history of zero weight, $Pr(|h_{SA}))=0$) due to the aforementioned bridging operators $B(t_2,t)=U_2\otimes I_{A}$ and $B(t,t_1)=U_1\otimes I_{A}$.

These considerations lead to a substantial difference between tracing out a subsystem at only a given time step and throughout all times of the whole history which requires to take into account evolution of the traced-out part. If we ``spatially'' trace out this component, this will lead to the reduced component $|h_{S})=[|\Phi_{1}\rangle\langle\Phi_{1}|]\odot[P_N]\odot[|\Psi_{0}\rangle\langle\Psi_{0}|]$ that has some non-zero probability of realization. This contradicts the fact that it is generated from a history $h_{SA}$ which cannot be realized. 
Finally, after tracing out the ancillary system from the history $|H^*_{SA})$, we get the following entangled history in the S system:
\begin{equation}\label{Hsprob}
|H_{S})=\gamma[|\Phi_{0}\rangle\langle\Phi_{0}|]\odot[P_N]\odot [|\Psi_{0}\rangle\langle\Psi_{0}|] +
\gamma[|\Phi_{1}\rangle\langle\Phi_{1}|] \odot [P_N]\odot[|\Psi_{1}\rangle\langle\Psi_{1}|],
\end{equation}
which leads to the same probability distribution for $P_N$ as in the case of the ABL formula applied to the two-time entangled state with a projective measurement between the times $t_1$ and $t_2$. In addition, Eq. \ref{Hsprob} clearly entails a non-local probability distribution characteristic of quantum entanglement.

\section{Isomorphism of the TSVF and the entangled history spaces}

The outline of this section relates to a comparative analysis of the entangled history space and the two-state vector space. Although both formalisms have deep differences rooted in their phenomenological interpretation of the behavior of wave functions traversing forward and backward in time, they lead to the same probability distributions for all measurement setups (as shown in Eq. \ref{Hsprob}) and therefore it should be possible to show that they can be made isomorphic under some conditions. However, the existing literature on these topics suggests that the state spaces generated in both formalisms are not isomorphic due to a lack of a proper inner product in the entangled histories approach.
To prove formally the isomorphism of two inner-product spaces, one therefore needs to equip the EH approach with a scalar product leading to the same results as the MTS approach. Noteworthily, these scalar products bring some physical information about the vectors representing temporal states that has a fundamental meaning as discussed below.

Let us consider then the behavior of scalar products in both formalisms starting with a simplified version of two-time states and histories. The following considerations can be easily extended to the multi-time case.
For the TSVF, a scalar product of a pair of vectors $|\Psi\rangle\rangle$ and $|\Phi\rangle\rangle$ in a space $\mathcal{M}=\mathcal{H}_{t_{2}}^{\dagger}\otimes\mathcal{H}_{t_1}$ with a basis $\mathcal{B}=\{\langle\phi_{i}^{2}||\phi_{j}^{1}\rangle\}$:
\begin{equation}
|\Psi\rangle\rangle=\sum_{ij}\alpha_{ij}\langle\phi_{i}^{2}||\phi_{j}^{1}\rangle
\end{equation}
\begin{equation}
|\Phi\rangle\rangle=\sum_{kl}\alpha_{kl}\langle\phi_{k}^{2}||\phi_{l}^{1}\rangle,
\end{equation}
is defined as follows \cite{Reznik}:
\begin{equation}
\langle\langle\Phi|\Psi\rangle\rangle=\sum_{ijkl}\alpha_{ij}\alpha^{*}_{kl}\langle\phi_{i}^{2}|\phi_{k}^{2}\rangle\langle\phi_{l}^{1}|\phi_{j}^{1}\rangle
=\sum_{ijkl}\alpha_{ij}\alpha^{*}_{kl}\delta_{ik}\delta_{jl}.
\end{equation}

An inner semi-definite product for history vectors $|\Psi)$ and $|\Phi)$ belonging to a space $\mathcal{E}=Proj(\mathcal{H}_{t_{2}})\odot Proj(\mathcal{H}_{t_1})$ in the EH representation:
\begin{equation}\label{Psi}
|\Psi)=\sum_{ij}\alpha_{ij}|\phi_{i}^{2}\rangle\langle\phi_{i}^{2}|\odot|\phi_{j}^{1}\rangle\langle \phi_{j}^{1}|
\end{equation}
\begin{equation}\label{Phi}
|\Phi)=\sum_{kl}\alpha_{kl}|\phi_{k}^{2}\rangle\langle\phi_{k}^{2}|\odot|\phi_{l}^{1}\rangle\langle \phi_{l}^{1}|,
\end{equation}
is defined as follows with application of the chain $K$-operator:
\begin{eqnarray}
(\Phi|\Psi)_K&=&Tr[K^\dag(|\Phi))K(|\Psi))]\nonumber\\
&=&Tr[\sum_{kl}\alpha_{kl}|\phi_{k}^{2}\rangle\langle\phi_{k}^{2}|U|\phi_{l}^{1}\rangle\langle \phi_{l}^{1}|]^\dag
[\sum_{ij}\alpha_{ij}|\phi_{i}^{2}\rangle\langle\phi_{i}^{2}|U|\phi_{j}^{1}\rangle\langle \phi_{j}^{1}|]\nonumber\\
&=&\sum_{ijkl}\alpha_{ij}\alpha^{*}_{kl}\langle\phi_{l}^{1}|U^\dag|\phi_{k}^{2}\rangle\langle\phi_{i}^{2}|U|\phi_{j}^{1}\rangle\delta_{ik}\delta_{jl} \nonumber\\
\end{eqnarray}
and so in general $\langle\langle\Phi|\Psi\rangle\rangle \neq(\Phi|\Psi)_K$, although for trivial bridging operator $U=\mathbb{I}$ they are equal. Thus, both spaces are not isomorphic if equipped with the aforementioned inner products. For making them isomorphic, we have to define a scalar product for the entangled history spaces which would be aligned with that of multi-time states.

\begin{defn}
A scalar product of a pair of history states $|\Psi)$ and $|\Phi)$ in a space $\mathcal{E}=Proj(\mathcal{H}_{t_{n}})\odot\cdots\odot Proj(\mathcal{H}_{t_{2}})\odot Proj(\mathcal{H}_{t_1})$ is defined as:
\begin{equation}
(\Phi|\Psi)_s\equiv Tr[|\Phi)^\dag|\Psi)].
\end{equation}
\end{defn}

This definition for the vectors $|\Phi)$ (\ref{Phi}) and $|\Psi)$ (\ref{Psi}) leads to the same result as in the case of TSVF:
$\langle\langle\Phi|\Psi\rangle\rangle=(\Phi|\Psi)_s=\sum_{ijkl}\alpha_{ij}\alpha^{*}_{kl}\delta_{ik}\delta_{jl}$. It is worth reminding that both vectors $|\Phi)$ and $|\Psi)$ are assumed to be normalized, thus, a probability amplitude for the realization of a particular history does not matter. What matters is a relative amplitude of realization ($|\Phi)$ in relation to $|\Psi)$) and as an implication, $K$-operators do not have to be engaged for such calculation. We can also briefly refer to a physical interpretation of $(\cdot|\cdot)_K$ and $(\cdot|\cdot)_s$. The K-operator folds any base history $|H)=[|\phi_{n}\rangle\langle \phi_{n}|]\odot[|\phi_{n-1}\rangle\langle \phi_{n-1}|]\odot\cdots\odot[|\phi_{0}\rangle\langle \phi_{0}|]$ to an operator $K(|H))=\alpha |x_n\rangle\langle x_0|$  and gives an amplitude of a process associated with the history. Consequently, the product $(\cdot|\cdot)_K$ does not capture the orthogonality relations between two histories for the intermediate time instances, what matter are the initial and final states with the realization amplitudes of the analyzed histories. In case of $(\cdot|\cdot)_s$, one gets a proper inner product for the history spaces which captures the orthogonality relations between analyzed histories for the intermediate states.

Consequently, we are ready to prove the isomorphism of both spaces equipped with appropriate inner products. The following theorem is constructed for two-time spaces but can be easily extended to a multi-time case.

\begin{thm}
A space $\mathcal{M}$ of multi-time state vectors equipped with a scalar product $\langle\langle\cdot|\cdot\rangle\rangle$ is isomorphic to a space $\mathcal{E}$ of entangled histories equipped with a scalar product $(\cdot|\cdot)_s$.
\end{thm}
\begin{proof}
To prove isomorphism of two spaces, one needs to show that there exists a bijective correspondence between vectors $\mathcal{M}\ni|\psi\rangle\rangle\leftrightarrow|\psi)\in\mathcal{E}$ such that the following conditions hold:\\
1. If $|\Psi\rangle\rangle\leftrightarrow|\Psi)$ and $|\Phi\rangle\rangle\leftrightarrow|\Phi)$, then $|\Psi\rangle\rangle+|\Phi\rangle\rangle \leftrightarrow|\Psi)+|\Phi)$. There holds a natural bijective correspondence that keeps also additivity for the vectors, i.e.:
\begin{eqnarray}
|\Psi_{t_{2}t_{1}}\rangle\rangle=\sum_{ij}\alpha_{ij}\langle\phi_{i}^{2}||\phi_{j}^{1}\rangle &\leftrightarrow& |\Psi_{t_{2}t_{1}})=\sum_{ij}\alpha_{ij}|\phi_{i}^{2}\rangle\langle\phi_{i}^{2}|\odot|\phi_{j}^{1}\rangle\langle \phi_{j}^{1}| \nonumber\\
|\Phi_{t_{2}t_{1}}\rangle\rangle=\sum_{ij}\beta_{ij}\langle\phi_{i}^{2}||\phi_{j}^{1}\rangle &\leftrightarrow& |\Phi_{t_{2}t_{1}})=\sum_{ij}\beta_{ij}|\phi_{i}^{2}\rangle\langle\phi_{i}^{2}|\odot|\phi_{j}^{1}\rangle\langle \phi_{j}^{1}| \nonumber
\end{eqnarray}
then $|\Psi_{t_{2}t_{1}}\rangle\rangle+|\Phi_{t_{2}t_{1}}\rangle\rangle \leftrightarrow|\Psi_{t_{2}t_{1}})+|\Phi_{t_{2}t_{1}})$.
This result can be easily reapplied for multi-time vectors: \\
2. If $|\Psi\rangle\rangle\leftrightarrow|\Psi)$, then $\lambda|\Psi\rangle\rangle\leftrightarrow\lambda|\Psi)$ for any $\lambda\in\mathcal{C}$ and $|\lambda|=1$ (where $\lambda$ represents a phase factor). Due to the correspondence: $|\Psi_{t_{2}t_{1}}\rangle\rangle=\sum_{ij}\alpha_{ij}\langle\phi_{i}^{2}||\phi_{j}^{1}\rangle \leftrightarrow |\Psi_{t_{2}t_{1}})=\sum_{ij}\alpha_{ij}|\phi_{i}^{2}\rangle\langle\phi_{i}^{2}|\odot|\phi_{j}^{1}\rangle\langle \phi_{j}^{1}|$
with complex $\alpha_{ij}$ and $\beta_{ij}$, there holds: $\lambda |\Psi_{t_{2}t_{1}}\rangle\rangle \leftrightarrow \lambda|\Psi_{t_{2}t_{1}})$.\\
3. For inner products, one gets $\langle\langle\Psi|\Phi\rangle\rangle=(\Psi|\Phi)_{s}$ with the assumed correspondence of the vectors. As shown for two-time states:
$\langle\langle\Phi|\Psi\rangle\rangle=(\Phi|\Psi)_s=\sum_{ijkl}\alpha_{ij}\alpha^{*}_{kl}\delta_{ik}\delta_{jl}$. The same result can be easily achieved for multi-time states and histories by extension.
\end{proof}







\section{Discussion} \label{sec4}
The Entangled Histories formalism and the Two-State-Vector formalism both try to capture the uniqueness of quantum histories  by allowing them to be superposed. This is in contrast with the Consistent Histories formalism which rather seeks for the cases where the Histories decohere. However, treating such superposition states and assigning probabilities to them were shown to be conceptually helpful \cite{WC3,Now2,CN,WC1,WC2}, especially in cases where the Consistent Histories formalism cannot do so \cite{Dis,CN,VaidCH}. Moreover, they have been recently shown to be very fruitful for studying various problem starting from quantum paradoxes \cite{Paradoxes} to the past of quantum particles \cite{Vaid1,Vaid2} and finally black hole information \cite{CN} and spacetime entropies/channels \cite{CotNew}.

In addition, the two approaches we discussed are intrinsically time-symmetric, reflecting the apparent reversibility in the microscopic world, therefore making them more appealing for us than other approaches. We have seen that apart from the backwards-in-time-evolution embedded within the TSVF, these approaches are quite similar in spirit. In fact, we showed that they can made isomorphic, when properly defining the inner product in the EH approach. We then saw that many notions usually ascribed to spatial quantum correlations apply to temporal correlations as well (see also \cite{Now2}), implying non-locality in time. While the latter makes these approaches particularly suitable for studying processes with indefinite causal order \cite{NewSandu}, quantum uncertainty prevents any violations of signaling in time and thus causality is always preserved \cite{ACE}.

As shown in the Appendix and in \cite{WC3} the above can be experimentally tested. For experimental demonstration purposes, sequential weak measurements were recently suggested as well \cite{GC}. The two frameworks emphasize a unique phenomenon in quantum mechanics, allowing histories to become entangled, thereby defying the classical notion of history and maybe even the classical notion of time itself.

Let us stress that the equivalence shown here
between the Entangled Histories and TSVF
raises a natural question about general probabilistic theories (GPT)
\cite{X,XX,XXX}. It is not impossible that
the entangled histories concept might be implemented within the GPT formalism due to its tensor-like structure
(especially because the dynamics in standard GPT is
well defined as a sort of mapping). In fact, a related time-symmetric approach has already been developed \cite{Ore1,Ore2}. If so, then an intriguing question would be whether there is any two-state
formalism and, perhaps, some analog of a wave function, going
beyond quantum mechanics.
This, however, requires more research, especially
as it may depend on the axioms chosen for the considered GPT.
Another possible area for further analysis might be related to
no-signaling boxes \cite{PRboxes}, e.g. finding their temporal analogs in a generalized entangled histories formalism. However, it seems that one would need
some extra structure, since apparently there is no reasonable dynamical structure
for this model at the moment.

\section*{Acknowledgements}

We wish to thank Yakir Aharonov, Jordan Cotler, Marek Czachor, Danko Georgiev, Sky Nelson and Sandu Popescu for helpful comments and discussions. M.N. and P.H. were supported by a grant from the John Templeton Foundation. E.C. was supported by the Canada Research Chairs (CRC) Program.  The opinions expressed in this publication are those of the authors and do not necessarily reflect the views of the John Templeton Foundation.

\section*{Appendix}
We present below a protocol for generation of $|\tau GHZ)$ state, i.e. a GHZ state of histories, that can be implemented experimentally with a Mach-Zehnder interferometer and a set of detectors \cite{WC3}:
\begin{equation}
|\tau GHZ)=\frac{1}{\sqrt{2}}([z^{+}]\odot[z^{+}]\odot[z^{+}]-[z^{-}]\odot[z^{-}]\odot[z^{-}])
\end{equation}
We start with a bipartite system at time $t_0$ consisting of a spin-$\frac{1}{2}$ particle P being in a state $|\phi_{0}\rangle=\frac{1}{\sqrt{2}}(|z^+\rangle+|z^-\rangle)$ ($|\phi_{0}\rangle=|x^+\rangle$) and a reference system R, consisting of three qubits in a state $|000\rangle$, which can be actually perceived as a clock for the process. Thus, at time $t_0$ the system PR is in a state (for states at each particular time, we write down the spatial state of the system in the $|\cdot\rangle$ notation):
\begin{equation}
t_0:\; |\Psi_{t_0}\rangle_{PR}=\frac{1}{\sqrt{2}}[(|z^+\rangle+|z^-\rangle)]|000\rangle
\end{equation}
Then, at a later time $t_1$ we act on the system with the $CNOT$ unitary operation where the control system is the particle and negation is performed on \textit{the first} qubit of the reference system R (the $CNOT$ operation changes the reference qubit if the controlled state is $|z^-\rangle$), basing on the state of the particle (we will repeat this action at time $t_2$ on the second qubit, and at $t_3$ on the third qubit):
\begin{equation}
t_1:\; |\Psi_{t_1}\rangle_{PR}=CNOT_{PR_1}\otimes \mathbb{I}_{R_{2}R_{3}}|\Psi_{t_0}\rangle_{PR}=\frac{1}{\sqrt{2}}|z^+\rangle|000\rangle+\frac{1}{\sqrt{2}}|z^-\rangle|100\rangle
\end{equation}
where $CNOT_{PR_1}$ acts on the particle and the first qubit of the reference system.\\
At time $t_2$ we act on the particle and the second qubit of the reference system achieving:
\begin{equation}
t_2:\; |\Psi_{t_2}\rangle_{PR}=CNOT_{PR_2}\otimes \mathbb{I}_{R_{1}R_{3}}|\Psi_{t_1}\rangle_{PR}=\frac{1}{\sqrt{2}}|z^+\rangle|000\rangle+\frac{1}{\sqrt{2}}|z^-\rangle|110\rangle
\end{equation}
Finally, at time $t_3$ we repeat this operation but on the particle and the third qubit of the reference system:
\begin{equation}
t_3:\; |\Psi_{t_3}\rangle_{PR}=CNOT_{PR_3}\otimes \mathbb{I}_{R_{1}R_{2}}|\Psi_{t_2}\rangle_{PR}=\frac{1}{\sqrt{2}}|z^+\rangle|000\rangle+\frac{1}{\sqrt{2}}|z^-\rangle|111\rangle
\end{equation}
After this step, we can measure the reference system in the computational basis $\{|000\rangle, |001\rangle, \ldots, |111\rangle\}$. If we measure the reference system projecting on $|000\rangle$ then particle has been in the history $[z^{+}]\odot[z^{+}]\odot[z^{+}]$. If we project it on $|111\rangle$, then the history of the particle (with which we correlate) has been in $[z^{-}]\odot[z^{-}]\odot[z^{-}]$.

Finally, if we project the reference system on $\frac{1}{\sqrt{2}}(|000\rangle-|111\rangle)$, we find that the particle has been in the history state $|\tau GHZ)=\frac{1}{\sqrt{2}}([z^{+}]\odot[z^{+}]\odot[z^{+}]-[z^{-}]\odot[z^{-}]\odot[z^{-}])$.

Alternatively, using MTS, the entangling role of the R system would be the same and also the sequence of measurements. The main difference is the alternating forward/backward evolution of the corresponding ket/bra states. This becomes clear when describing the above evolution in the MTS formalism as $\frac{1}{\sqrt{2}}\left ( _{t_f}\langle x^-|_{t_3}\langle z^+||z^+\rangle_{t_2~t_1}\langle z^+||x^+\rangle_{t_0}+_{t_f}\langle x^-|_{t_3}\langle z^-||z^-\rangle_{t_2~t_1}\langle z^-||x^+\rangle_{t_0} \right )$.
We note that although the initial and final states are orthogonal there is no conflict with the postulates of the MTS as projective measurements were performed during the time evolution of the system.

\end{document}